\documentclass[12pt]{article}
\usepackage[all]{xy}

\usepackage{epic}
\usepackage{amsmath}[1996/11/01]
\usepackage{amssymb,amsthm,amsfonts,latexsym,epsfig,graphics}
 \def\beql#1#2\eeql{\begin{equation}\label{#1}#2\end{equation}}

\advance \textheight by -1 \topmargin
\advance \textheight by -1 \headheight
\advance \textheight by -1 \headsep
\advance \textheight by -2 \footskip
\vsize = \textheight
\hsize = \textwidth
\newcommand{\rscp}{\mbox{\bf {\sf )}}} 
\newcommand{\lscp}{\mbox{\bf {\sf (}}} 

\DeclareMathOperator{\Gal}{Gal}
\DeclareMathOperator{\Rk}{Rk}
\DeclareMathOperator{\Char}{char}

\DeclareMathOperator{\Trace}{Trace}
\DeclareMathOperator{\Beta}{{\mathfrak B}}
\DeclareMathOperator{\Tau}{{\mathcal T}}
\DeclareMathOperator{\tr}{{\sf T}}
\DeclareMathOperator{\mtr}{{\sf -T}}
\DeclareMathOperator{\trace}{trace}

\DeclareMathOperator{\diag}{diag}

\DeclareMathOperator{\GL}{GL}
\DeclareMathOperator{\Aut}{Aut}

\DeclareMathOperator{\id}{id}

\newtheorem{theorem}{Theorem}[section]

\newtheorem{prop}[theorem]{Proposition}

\newtheorem{cor}[theorem]{Corollary}

\newtheorem{rem}[theorem]{Remark}
\newtheorem{lemma}[theorem]{Lemma}
\newtheorem{proposition}[theorem]{Proposition}
\newtheorem{definition}[theorem]{Definition}

\newcommand{\Z}{{\mathbb{Z}}}

\newcommand{\F}{{\mathbb{F}}}
\newcommand{\N}{{\mathbb{N}}}

\newcommand{\MRD}{\rm MRD}

\renewcommand{\em}{\sf}

\begin{document}
\begin{center}
{\Large {\bf On self-dual $\MRD$ codes}} \\
\vspace{1.5\baselineskip}
{\em Gabriele Nebe\footnote{
Lehrstuhl D f\"ur Mathematik, RWTH Aachen University,
52056 Aachen, Germany, 
 nebe@math.rwth-aachen.de} and Wolfgang Willems\footnote{Otto-von-Guericke Universit\"at, Magdeburg, Germany and Departamento de Matem{\'a}ticas, Universidad del Norte, Barranquilla, Colombia, willems@ovgu.de} }
\end{center}

{\sc Abstract.}
{\small 
We investigate self-dual MRD codes. In particular we prove that a Gabidulin code in $(\F_q)^{n\times n}$ is equivalent to a self-dual code
if and only if its dimension is $n^2/2$, 
$n \equiv 2 \pmod 4$, and $q \equiv 3 \pmod 4$. On the way we determine the full automorphism group of Gabidulin codes in $(\F_q)^{n\times n}$.
\\[1ex]
}

\section{Introduction.}


Following Delsarte \cite{Delsarte}
a rank metric code is a set ${\mathcal C} \subseteq k^{m \times n}$ 
of $m\times n$ matrices over a field $k$. 
The distance between two matrices $A,B \in k^{m \times n}$ is 
defined as $d(A,B):= \Rk (A-B)$, i.e. the rank of the difference of $A$ and $B$.
 As usual we denote 
by
$$d({\mathcal C}) := \min \{ d(A,B) \mid A, B\in {\mathcal C}, A\neq B  \}$$
the minimum distance of
 ${\mathcal C}$. 
The dual code of ${\cal C}$   is 
$${\mathcal C}^{\perp } = \{ X\in k^{m \times n} \mid \lscp C,X \rscp := 
\trace(CX^{\tr} ) = 0 \mbox{ for all }
C \in {\mathcal C} \} $$
where $X^{\tr}$ is the transpose  and $\trace(X)$  the trace of the matrix $X$.
Clearly, ${\mathcal C}^{\perp }$ is always a $k$-linear code, i.e. a subspace 
of the $k$-vector space $k^{m  \times n}$. 

Throughout the paper we assume that $m\geq n$, so our
matrices have at least as many rows as columns. 
We will also assume that ${\mathcal C}$ is a linear code. 
If ${\mathcal C} \leq k^{m \times n}$ has dimension $\ell $ and minimum distance $d$,
then $d\leq n-\ell/m +1$ (see \cite[Theorem 5.4]{Delsarte}, \cite[Theorem 8]{Ravagnani}). Codes where equality holds are called 
$\MRD$ codes 
(maximum rank distance codes). 
By 
 \cite[Theorem 5.4]{Delsarte}
 the dual of an $\MRD$ code is again an $\MRD$ code 
(see also  \cite[Corollary 41]{Ravagnani}).

In this note we investigate self-dual $\MRD $ codes, i.e. $\MRD$ codes
${\mathcal C}$ with ${\mathcal C} = {\mathcal C}^{\perp} $.
As $\dim ({\mathcal C} ) + \dim ({\mathcal C}^{\perp }) = \dim (k^{m\times n}) = mn $ 
a self-dual $\MRD $ code 
${\mathcal C} \leq k^{m \times n}$ 
 with $m\geq n$ has dimension
$\frac{mn}{2}$ and minimum distance $d({\mathcal C}) = \frac{n}{2}+1$.

Section \ref{AUTO} investigates 
which rank distance preserving linear automorphisms of $k^{m\times n}$ 
stabilise the inner product $\lscp - , - \rscp$ defined by $\lscp A,B \rscp = \trace (AB^{\tr})$. 

This inner product $\lscp - , - \rscp $ is the standard inner
product if we identify $k^{m\times n} $ with $k^{1\times mn}$. 
If $\Char (k)  =2 $, then self-dual codes in $k^{1\times mn}$
always contain the all-ones vector. So self-dual rank metric
codes contain the all-ones matrix $J \in \{ 1 \}^{m\times n}$ of rank 1.
This implies that
 there are no self-dual $\MRD$ codes over fields of characteristic 2
(see Theorem \ref{char2}).
In Section \ref{equivalent} we give in odd characteristic a handy criterion to prove if a given rank metric code in $k^{m\times n}$ is
equivalent to a self-dual code (see Theorem \ref{eqsd}).

In the rest  of the paper we study $\MRD$ codes in $k^{n\times n}$ where
$k$ is a finite field.
In case $n=2$  all self-dual $\MRD $ codes are classified in
Section \ref{SectionSelfDualMRD}: They exist 
if and only if $-1$ is not a square in $k$. 

The most well-studied examples of $\MRD $ codes are the 
Gabidulin codes (\cite{Gabidulin}, \cite{Delsarte}). 
Section \ref{GABI} treats Gabidulin codes 
 of full length $n$, i.e. $n=[K:k]$ is the degree of the field extension, as $k$-linear
subspaces of dimension $\ell n$  of $k^{n\times n}$. 
We determine the $k$-linear automorphism group of these codes
(see Corollary \ref{fullaut}) and
show that such a Gabidulin code  
is equivalent to a self-dual code if and only if $n \equiv 2 \pmod{4}$,
$\ell = n/2$, 
and $-1$ is not a square in $k$ (see Theorem \ref{gabisd}).

If $-1$ is a square in $k$ or $n$ is a multiple of 4, we do not have any 
examples of self-dual $\MRD $ codes in $k^{n\times n}$. 
 Note that
according to \cite{MorrisonThesis} there are 5 equivalence classes of
self-dual $\MRD$ codes in $\F_5^{4\times 2}$.

\section{Automorphisms preserving the inner pro\-duct.} \label{AUTO}

The rank distance preserving automorphisms of $k^{m \times n}$ are 
$$ \kappa _{X,Y,Z,\sigma}:   A\mapsto XA^{\sigma }Y+Z   \mbox{ with }
X\in \GL_m(k),Y\in \GL_n(k), Z\in k^{m \times n}, \sigma\in \Aut(k)  $$ or $$
\tau _{X,Y,Z,\sigma}:    A\mapsto XA^{\tr,\sigma }Y+Z  \mbox{ with } X,Y\in \GL_n(k), Z\in k^{m \times n} , \sigma \in \Aut(k) \ 
(\mbox{if } m=n)  $$
and these are $k$-linear, if and only if $Z=0$ and $\sigma = \id $ (see \cite{Wan}, Theorem 3.4). 
If $m=n$, then the $\tau_{X,Y} := \tau _{X,Y,0,{\rm \small id}}$ are called {\em improper} and
the $\kappa_{X,Y}:=\kappa _{X,Y,0, {\rm \small id} }$ {\em proper} automorphisms.
The group of proper automorphisms is a normal subgroup of index
2 in the full automorphism group of $k^{n\times n}$. \\

\begin{rem}\label{dualequiv} {\rm
For $C,D\in k^{m\times n}$, $X\in \GL_m(k)$, and $Y\in \GL_n(k)$  we compute
$$\trace(XCY(X^{\mtr} D Y^{\mtr}) ^{\tr} ) = 
\trace(XCYY^{-1} D^{\tr} X^{-1})  = \trace(CD^{\tr}) $$ 
where the second equality follows from the fact that the trace is invariant under conjugation.
This 
shows that 
$$\kappa _{X,Y}  ({\mathcal C}) ^{\perp } =
\kappa _{X^{\mtr}, Y^{\mtr}} ({\mathcal C}^{\perp} ) .$$}
\end{rem}

By $I_n \in k^{n \times n}$ we denote the identity matrix.
The following result, which we do not need in the rest of this note, is of its own interest.

\begin{prop}
The rank distance preserving linear automorphisms of $ k^{m \times n}$ that 
preserve the inner product $\lscp -,- \rscp $  are exactly  the maps
$\kappa _{X,Y} $  and, if $m=n$,  $\tau _{X,Y}$ where
$$ X^{\tr}X = a I_m, \ YY^{\tr} = a^{-1}I_n \ \mbox{and} \ a \in  k^{\times }\mbox{, the multiplicative group of } k .$$
\end{prop} 
\begin{proof}
The maps $\kappa _{X,Y} $  and, if $m=n$, also $\tau _{X,Y}$, where $X \in \GL(m,k)$ and $Y \in \GL(n,k)$, are the linear automorphisms
which preserve the rank distance on $k^{m \times n}$. 
If $n=m$ then $\tau _{I_n,I_n}$ is an improper automorphism that 
preserves the inner product and $\tau _{X,Y} \tau _{I_n,I_n} = \kappa _{X,Y}$,
 so it is enough to deal with proper automorphisms.  \\
Suppose that $\kappa _{X,Y} $ preserves the inner product. 
Thus, for $A,B \in k^{m \times n}$ we have
$$ \trace(AB^{\tr}) = \trace(XAY(XBY)^{\tr}) = \trace((X^{\tr}XA)(YY^{\tr}B^{\tr})).$$
We put $(x_{ij}) = X^{\tr}X$ , $(y_{ij}) = YY^{\tr}$ and denote
 by $E_{ij} \in k^{m \times n}$ the matrix which has a $1$ at position $(i,j)$ and $0$'s elsewhere.
Hence, for $A=E_{ij}$ and $B= E_{lh}$, we obtain the equation
$$ x_{li}y_{jh} = \trace((X^{\tr}XE_{ij})(YY^{\tr}E_{hl})) = \trace(E_{ij}E_{hl}) = \delta_{il}\delta_{jh}.$$
In particular,
$$  y_{jj} = x_{ii}^{-1}$$
for all $i,j$. 
Furthermore, $ \delta_{jh} = x_{ii}y_{jh}$ forces $y_{jh} =0$ for $j \not= h$. By the same argument we get $x_{li}= 0$ for $l \not= i$.
This shows $X^{\tr}X = aI_m$ with $a \in k^{\times }$ and $ YY^{\tr} = a^{-1}I_n$.
\\[1ex]
Since the given maps obviously preserve both the rank distance and the inner product the proof is complete.
\end{proof}

\begin{rem} {\rm
As one referee pointed out, 
the same proof shows that the linear rank distance preserving automorphisms that 
yield similarities for the inner product  $\lscp -,- \rscp $ (which is enough to preserve
the notion of duality) are the maps 
$\kappa _{X,Y} $  and, if $m=n$,  $\tau _{X,Y}$ where
$ X^{\tr}X = a I_m, \ YY^{\tr} = bI_n ,  a,b \in  k^{\times }$.
}
\end{rem}

\section{Self-dual $\MRD$ codes}  \label{SectionSelfDualMRD}

Surprisingly, in characteristic $2$ self-dual $\MRD $ codes 
in $k^{m  \times n}$ do not exists. 
This follows immediately from
the following easy, but crucial result, since a self-dual $\MRD$ code in $k^{m  \times n}$ has at least minimum distance $2$.

\begin{theorem} \label{char2}
Assume that $\Char(k) =2$ and let ${\mathcal C} \subseteq {\mathcal C}^{\perp } \leq 
k^{m \times n}$
be a self-orthogonal code. Then the all-ones matrix $J$ is in ${\mathcal C}^{\perp }$. In particular,  $d({\mathcal C}^{\perp } ) = 1$. 
\end{theorem}

\begin{proof}
All elements $A\in {\mathcal C}$ satisfy 
$$0=\lscp A,A \rscp = 
\sum _{i=1}^m \sum _{j=1}^n A_{ij}^2 = 
(\sum _{i=1}^m \sum _{j=1}^n A_{ij})^2 = \lscp A , J \rscp ^2 $$
where $J $ is the all-ones matrix, which is of rank 1.
So $J\in {\mathcal C}^{\perp }$ satisfies $d(0,J) = 1$.
\end{proof}

In contrast to the characteristic $2$ case self-dual $\MRD$ codes may exist if $\Char(k)$ is odd. To see that we characterize all
self-dual $\MRD$ codes $\mathcal C$ in $k^{2 \times 2}$ where $k=\F_q$ is the
finite field with $q$ elements.
 Since $d({\mathcal C}) = 2$ and $\dim ({\mathcal C}) = 2$,
the projection on the first row 
$$\pi : {\mathcal C} \to k^{1 \times 2}, \quad A \mapsto (a_{11}, a_{12}) $$
is an isomorphism and ${\mathcal C}$ has a unique basis of the 
form 
$$ A = \left(\begin{array}{cc} 1 & 0 \\ a & b \end{array} \right), \ 
 B = \left(\begin{array}{cc} 0 & 1 \\ c & d \end{array} \right)  $$
with $a,b,c,d\in k$.

\begin{prop} \label{dim2}
${\mathcal C} = \langle A,B \rangle $ is a self-dual $\MRD$ code
if and only if  the following two conditions hold true.
\begin{itemize} \item[\rm (i)] $-1 \not\in (k^{\times })^2 $, i.e.,  $q \equiv 3 \pmod{4}$. 
\item[\rm (ii)]
$a^2+b^2 = -1 \mbox{ and } (c,d) \in \{ (-b,a), (b,-a ) \}. $
\end{itemize} 
\end{prop}

\begin{proof}
Assume that ${\mathcal C} = \langle A,B \rangle $ is a self-dual code.
Then 
$\lscp A,A \rscp = \lscp A,B \rscp = \lscp B,B \rscp = 0$ yields the equations
$$a^2+b^2 +1=  c^2+d^2+1 =  ac+bd = 0 .$$
The ideal in $\Z[a,b,c,d]$ generated by these three polynomials contains the element
$$ a^2( c^2+d^2+1) -d^2(a^2+b^2+1) +(bd-ac)(ac+bd) = a^2-d^2 = (a+d)(a-d).$$
We therefore conclude that $a=\pm d$ and similarly $b=\pm c$. Thus 
condition (ii) is equivalent to ${\mathcal C}$ being self-dual. 
Moreover ${\mathcal C}$ is an $\MRD$ code, if all non-zero matrices in ${\mathcal C}$ 
have determinant $\neq 0$, so if and only if $b\neq 0$ and 
$$\det (A +xB ) = \left\{ \begin{array}{rl} (x^2+1) b, & \mbox{if} \  (c,d) = (-b,a) \\ 
-(x^2+2\frac{a}{b}x-1)b, & \mbox{if} \ (c,d) = (b,-a) \end{array} \right. $$
is an irreducible polynomial in $k[x]$. 
Using the fact that $a^2+b^2=-1$ we see that in both cases this leads to the 
condition that $-1$ is not a square in $k$, so $q \equiv 3 \pmod{4}$. Thus we have  (i). 
With the same computations 
as above we see that the conditions in (i) and (ii) lead to a self-dual $\MRD$ code.
\end{proof}

It is easy to see that all these codes are pairwise equivalent and
that they are equivalent to  Gabidulin codes of full length. 
So Proposition \ref{dim2} may be seen as a special case of Theorem \ref{gabisd} below.

\section{A criterion to be equivalent to a self-dual code} \label{equivalent}

\begin{definition} {\rm
Two linear rank metric codes ${\mathcal C} $ and ${\mathcal D} \leq k^{m\times n}$ are called {\em properly equivalent}, if 
there are $X \in \GL_m(k)$, $Y\in \GL_n(k)$ such that 
${\mathcal D}= X {\mathcal C} Y $. }
\end{definition} 

Note that proper equivalence is the usual notion of linear equivalence 
for  $m\neq n$. Only for $m=n$ the proper equivalences form a subgroup of
index 2 in the group of linear equivalences; the latter also
includes transposition of matrices (see Section \ref{AUTO}).

\begin{lemma} \label{qf} 
Let $k$ be a finite field of odd characteristic 
and let $A\in k^{n\times n} $ be a symmetric matrix of full rank. Then
there is a matrix $X\in \GL_n(k)$ such that $A= X X^{\tr}$ if and only if 
$\det (A) \in (k^{\times })^2 $. 
\end{lemma}

\begin{proof}
Regular quadratic forms 
over finite fields of odd characteristic are classified by 
their dimension and their determinant 
(see for instance \cite[Chapter 2, Theorem 3.8]{Scharlau}).
In particular a quadratic form with Gram matrix $A\in \GL_n(k)$
 is equivalent to the
standard form with Gram matrix $I_n$ if and only if $\det(A) $ is a square.
\end{proof}

\begin{theorem}\label{eqsd}
Let $k$ be a finite field of odd characteristic and let
${\mathcal C} \leq k^{m\times n}$ be a linear rank metric code.
Then 
${\mathcal C}$ is properly equivalent to a self-dual code if and only if
there are symmetric matrices $A = A^{\tr}\in k^{m\times m}$ and 
$B=B^{\tr} \in k^{n\times n}$ such that $\det(A) $, 
$\det(B) \in (k^{\times })^2$ are non-zero squares with
$${\mathcal C}^{\perp} = A {\mathcal C} B .$$
\end{theorem}

\begin{proof}
Assume that there are $X \in \GL_m(k)$, $Y\in \GL_n(k)$ such that 
${\mathcal D}:= X {\mathcal C} Y = {\mathcal D}^{\perp }$. 
Then  for all $C_1,C_2\in {\mathcal C}$ we have 
$$ 0 = \trace (X C_1 Y (X C_2 Y)^{\tr} ) = 
\trace (X C_1 Y Y^{\tr} C_2 ^{\tr}  X^{\tr} ) = 
\trace (X^{\tr} X C_1 Y Y^{\tr} C_2 ^{\tr}  )  $$ 
Put $A:= X^{\tr } X$ and $B:= Y Y^{\tr} $.
Then $A$ and $B$ are symmetric of square determinant and
${\mathcal C}^{\perp } = A {\mathcal C} B $. 
\\
On the other hand assume that there are $A,B$ as stated in the theorem.
According to  Lemma \ref{qf} there are 
$X \in \GL_m(k)$, $Y\in \GL_n(k)$ such that 
$A=X^{\tr} X $, $B=Y Y^{\tr}$. 
The same computation as above 
shows that 
$X{\mathcal C} Y $ is a self-dual code.
\end{proof}

\section{Gabidulin codes in $k^{n \times n}$} \label{GABI}

We keep the assumption that $k=\F_q $ is a finite field, 
but allow $\Char(k)$ to be arbitrary (even or odd).
Let $K:= \F_{q^n}$ be the degree $n$ extension field of $k$.
For $\alpha \in K$ and $0\leq i \leq n-1$ we define 
$\alpha ^{[i]} := \alpha ^{q^i} $ to be the image of $\alpha $ under
the $i$-th iteration of the Frobenius automorphism of $K/k$ and 
$\Trace _{K/k} (\alpha ):= \sum _{i=0}^{n-1} \alpha ^{[i]} $. 
For a 
 $k$-basis $\Beta :=(\beta _1,\ldots , \beta _n)$ of $K$  the dual 
basis $\Beta ^*:=(\beta _1^*,\ldots , \beta _n^*)$ is defined by the 
property that $\Trace _{K/k} (\beta _i \beta _j^*) = \delta _{ij} $.
If $\beta _i^* = \beta _i$ for all $i$, then the basis $\Beta $ is 
called a self-dual basis. 
Note that a dual basis always exists, but a self-dual basis exists if and only if $q$ is even or both $q$ and $n$ are odd (see \cite{LS}).
Let $\Tau _{\Beta }:= (\Trace _{K/k} (\beta_i \beta_j))_{i,j=1,...,n} $ denote the Gram matrix 
of the trace bilinear form 
$(\alpha,\beta)\in K\times K \mapsto \Trace _{K/k} (\alpha \beta ) \in k $
 with respect to the basis $\Beta $.
Then $\Tau _{\Beta }$ is the base change matrix between $\Beta $ and its 
dual basis $\Beta ^*$, because if $\beta _i = \sum _{m=1}^n a_{mi} \beta _m^* $ 
with $a_{mi} \in k$ for all $m$, then 
$$(\Tau _{\Beta }) _{ji} = 
\Trace _{K/k} (\beta_j \beta_i) = \sum _{m=1}^n a_{mi} \Trace _{K/k} (\beta _j \beta_m^* ) 
= a_{ji} .$$
In the notation of the next definition $\Tau _{\Beta } = \epsilon _{\Beta ^*} (\Beta )$.

\begin{definition} {\rm 
Let $\Beta = (\beta _1,\ldots , \beta _n) \in K^n$ be a $k$-basis of $K$
and define the map $\epsilon _{\Beta } : K^{1\times n}\to k^{n\times n} $ by
$$\epsilon_{\Beta } (\alpha _1,\ldots , \alpha _n) := (a_{ij}) \in k^{n\times n} \  \mbox{ if } 
\alpha _j = \sum _{i=1}^n a_{ij} \beta _i .$$
For $\alpha \in K$ and $\ell \in \N_0$ 
we also put $\alpha \Beta := (\alpha \beta_1,\ldots , \alpha \beta _n) \in K^n$ and 
$\alpha ^{[\ell ]} := \alpha ^{(q^\ell )} $ respectively 
$\Beta ^{[\ell ]} := (\beta _1^{[\ell ]} ,\ldots , \beta _n^{[\ell ]} ) $. }
\end{definition}

\begin{lemma}\label{basechange}
For any $\alpha \in K$ we have 
$$\epsilon _{\Beta } ( \alpha \Beta ) ^{\tr } 
= \epsilon _{\Beta ^*} (\alpha \Beta ^*) = \Tau _{\Beta } 
\epsilon _{\Beta }(\alpha \Beta ) \Tau _{\Beta }^{-1} .$$
\end{lemma}

\begin{proof}
Let $B:= \epsilon _{\Beta }(\alpha \Beta ) $ and $\Tau _{\Beta } =: \Tau $.  If we denote the entry of a matrix $A$ at position $(i,j)$ by $A_{ij}$, then
$$\beta _j = \sum _{i=1}^n \Tau _{ij} \beta _i^*, \ 
\beta _j^* = \sum _{i=1}^n (\Tau ^{-1})_{ij} \beta _i \mbox{, and } 
\alpha \beta _i = \sum _{j=1}^n B_{ji} \beta _j .$$
We compute
$$ \begin{array}{rcl} \alpha \beta _{r }^* & = & 
\alpha \sum _{i=1}^n (\Tau ^{-1})_{i{r }}  \beta _i \\[1ex]  & = &
\sum _{i=1}^n (\Tau ^{-1})_{i{r }} (\alpha \beta _i ) \\[1ex]  &=&  
\sum _{i=1}^n (\Tau ^{-1})_{i{r }}  \sum _{j=1}^n B_{ji} \beta _j \\[1ex] & = &
\sum _{i=1}^n (\Tau ^{-1})_{i{r }}  \sum _{j=1}^n B_{ji} \sum _{m=1}^n \Tau _{mj} \beta _m^*  \\[1ex]
& =& \sum _{m=1}^n (\Tau  B \Tau ^{-1} ) _{m,r } \beta _m^* ,
\end{array} $$
which proves the second equality.
To see the first equality let $C:=\epsilon _{\Beta ^*} (\alpha^{-1}  \Beta ^*) $,
so $ \alpha ^{-1} \beta _r^* = \sum _{s=1}^n C_{sr} \beta _s^*$.
Note that $BC^{\tr} = I_n$ since
$$ \begin{array}{rcl} \delta_{ir} & = & \Trace_{K/k} (\alpha \beta _i \alpha ^{-1} \beta _r^* ) \\[1ex] & = &
\sum _{j=1}^n B_{ji} \sum _{s=1}^n C_{sr} \Trace_{K/k} (\beta _j \beta_s^*)  \\[1ex]
& = & \sum _{j=1}^n B_{ji} C_{jr} =  (BC^{\tr} ) _{ir} . \end{array} $$
So $B^{\tr} = C^{-1} = (\epsilon _{\Beta ^*} (\alpha^{-1}  \Beta ^*)  )^{-1} =
\epsilon _{\Beta ^*} (\alpha  \Beta ^*) $. The last equality follows from the fact that the matrix  $\epsilon _{\Beta ^*} (\alpha  \Beta ^*) $
describes the $k$-linear map induced by the multiplication of $\alpha$ on $K$ with respect to the basis $\Beta^*$.
\end{proof}

\subsection{Automorphisms of Gabidulin codes.} 

In this section we  determine the 
automorphism group  of 
Gabidulin codes of full length $n=[K:k]$. 
To obtain a nice description in terms of matrices we use
a normal basis
$\Gamma := (\gamma,\gamma ^{[1]},\ldots, \gamma ^{[n-1]} ) $
 of $K$ over $k$.
Define $\Tau :=\Tau _{\Gamma } = \epsilon_{\Gamma ^*}(\Gamma ) \in k^{n\times n} $ to 
be the Gram matrix of the trace bilinear form with respect to $\Gamma $ and 
let
$$A:= \left( \begin{array}{cccc} 
0 & \ldots &  0 & 1 \\
1 & 0 & \ldots & 0  \\
0 & \ddots & \ddots & \vdots \\
0 & \ldots & 1 & 0 \end{array} \right) = \epsilon _{\Gamma }(\Gamma ^{[1]}).$$

By direct matrix computations we obtain the following elementary 
properties of these matrices.

\begin{rem} \label{vert} 
\begin{itemize}
\item[(i)] $A\Tau  = \Tau A $.
\item[(ii)] For $1\leq j\leq n-1$ we have $\Gamma ^{[j]} = \Gamma A^j $. 
\end{itemize}
\end{rem} 


\begin{definition} {\rm 
For $1\leq \ell \leq n$ the Gabidulin code
${\mathcal G}_{\ell ,\Gamma } \leq k^{n \times n }$ 
is the $k$-linear code 
$${\mathcal G}_{\ell ,\Gamma } = \langle \epsilon_{\Gamma } (\gamma _i \Gamma ^{[j]} ) \mid  1\leq i \leq n, 0\leq j\leq \ell -1 \rangle .$$ 
Let ${\mathcal K} := {\mathcal G}_{1,\Gamma }$. }
\end{definition}

\begin{lemma} \label{notation} 
\begin{itemize}
\item[\rm (i)]
${\mathcal K}$ is an $n$-dimensional subalgebra of $k^{n\times n}$ isomorphic to $K= \F_{q^n}$.
\item[\rm (ii)] For any $B \in {\mathcal K} $ we have  $AB A^{-1} = B^q $.
In particular $A{\mathcal K} = {\mathcal K} A$ as a set.
\item[\rm (iii)] 
The normalizer in $\GL_n(k)$ of ${\mathcal K}^{\times }$ is 
the semidirect product of ${\mathcal K}^{\times }$ and the cyclic group $\langle A \rangle $
of order $n$.
\item[\rm (iv)] $\trace (B A^{\ell}) = 0 $ for all $B\in {\mathcal K}$ and  all
$1\leq  \ell \leq n-1$.
\item[\rm (v)] 
The full matrix ring 
$$k^{n\times n}  = {\mathcal K} \oplus {\mathcal K} A 
\oplus \ldots \oplus {\mathcal K} A^{n-1}$$ 
is a cyclic algebra.
So for all $X\in k^{n\times n}$ there are unique 
$x_i \in {\mathcal K}$ such that $X = \sum _{i=0}^{n-1} x_i A^i$. 
\item[\rm (vi)]
For $\ell \geq 1$ 
$${\mathcal G}_{\ell , \Gamma } = {\mathcal K} \oplus {\mathcal K} A \oplus \ldots \oplus {\mathcal K} A^{\ell -1 } .$$ 
\end{itemize}
\end{lemma}

\begin{proof}
(i) The map $K\to {\mathcal K}$, $\alpha \mapsto \epsilon _{\Gamma }( \alpha \Gamma ) $ 
is an isomorphism of $k$-algebras. 
\\
(ii) 
We use the isomorphism above to write $B = \epsilon _{\Gamma }(\beta \Gamma )$ for 
some $\beta \in K$ and recall 
that $A_{ij} = \delta _{i,(j+1)} $.
We show that $AB =B ^q A $ for all $B \in {\mathcal K}$.
By definition we have that 
$$\beta \gamma ^{[j]}  = \sum _{i=0}^{n-1} B_{ij}\gamma ^{[i]} ,\ (AB)_{(i+1)j} = B_{ij}, \mbox{ and } (B^q A) _{ij} = (B^q)_{i(j+1)} .$$
Therefore we compute  for all $j=0,\ldots , n-1$
$$\begin{array}{l} 
\sum _{i=0}^{n-1} (B^qA)_{ij} \gamma ^{[i]}  = 
\sum _{i=0}^{n-1} (B^q)_{i(j+1)} \gamma ^{[i]}  = 
\beta ^q \gamma ^{[j+1]} = (\beta \gamma ^{[j]} )^{[1]}   = 
(\sum _{i=0}^{n-1} B_{ij} \gamma ^{[i]} ) ^{[1]}  = \\
\sum _{i=0}^{n-1} B_{ij} \gamma ^{[i+1]}   =    
\sum _{i=0}^{n-1} (AB)_{(i+1)j} \gamma ^{[i+1]}   =     
\sum _{i=0}^{n-1} (AB)_{ij} \gamma ^{[i]}   .\end{array} $$
So the $j$-th column of $B^qA$ and $AB$ coincide. 
\\
(iii) 
This is well-known and widely used in geometry and group theory,
see for instance \cite{Huppert}, Kap. II, Satz 7.3.
\\
(iv)  We consider the matrices in $K^{n\times n}$. 
Take any primitive element $\alpha \in K$. 
Then $C:=\epsilon _{\Gamma }(\alpha \Gamma )\in \GL_n(K)$ has $n$ 
distinct eigenvalues $\alpha , \alpha ^{[1]} , \ldots , 
\alpha ^{[n-1]}$, the roots of the minimal polynomial of $\alpha $ 
over $k$. 
In particular there is a matrix 
$X \in \GL_n(K)$
such that $X^{-1}CX = \diag (\alpha , \alpha ^{[1]} , \ldots , 
\alpha ^{[n-1]})$. 
As $A C A^{-1} = C^q$ (by (ii)) also 
$(X^{-1}AX) (X^{-1} C X) (X^{-1} A X)^{-1} = (X^{-1}CX)^q$, 
so $X^{-1} A X $ cyclically permutes the eigenspaces of 
$X^{-1} C X $. More precisely there are $a_i\in K$ such that 
$$(X^{-1} A X )_{ij} =   \left\{ \begin{array}{ll} a_i & j = i+1 \\ 0 & \mbox{ otherwise } \end{array}\right. $$
where as usual the indices are taken modulo $n$. 
Because $k[C] = {\mathcal K}$, 
any $B \in {\mathcal K}$ is a polynomial in $C$ and hence 
$X^{-1} B X $ is a diagonal matrix. 
So for any $1\leq i \leq n-1$ the 
matrix $X^{-1} BA^i X $ is  monomial with no non zero entries on the diagonal,
because it induces the fixed point free permutation $(1,2,\ldots,n)^i$ 
on the eigenspaces of $X^{-1} C X $. 
In particular its trace is 0. As the trace is invariant under 
conjugation we also get 
$\trace (BA^i) = \trace ( X^{-1} BA^i X ) = 0 $. \\[1ex]
(v)  Suppose that $\sum_{i=0}^{n-1}B_iA^i =0$ where $B_i \in {\mathcal K}$.
Note that 
$B_iA^i = \epsilon_\Gamma(\beta_i \Gamma^{[i]})$.
Thus we obtain
$  \epsilon_\Gamma(\sum_{i=0}^{n-1} \beta_i \Gamma^{[i]}) = 0,$ hence  $ \sum_{i=0}^{n-1} \beta_i \Gamma^{[i]} =(0,\ldots , 0)$ since $\epsilon_\Gamma$ is injective.
By \cite[Chapter 3, Lemma 3.50]{LN}, the $\Gamma^{[i]}$ are linearly independent over $K$, hence $\beta_i =0$ for all $i$. This proves that the
right hand  side of the equation in (v) is a direct sum. The equality follows by comparing dimensions.\\
(vi) This follows immediately from (v) using the definition of ${\mathcal G}_{\ell ,\Gamma }$.
\end{proof}

We are now ready to determine the automorphism group of ${\mathcal G}_{\ell ,\Gamma }$ for all $\ell $. 
Clearly ${\mathcal G}_{0, \Gamma } := \{ 0 \}$ and
 ${\mathcal G}_{n,\Gamma } = k^{n\times n}$ are fixed by all 
linear equivalences (introduced in Section \ref{AUTO}).
Also for the other Gabidulin codes 
there are certain obvious matrices $(X,Y)\in \GL_n(k) \times \GL_n(k)$, 
so that 
$X {\mathcal G}_{\ell ,\Gamma } Y = {\mathcal G}_{\ell ,\Gamma }$ (see for instance \cite{Morrison}):
\\
For notational convenience we put
 ${\mathcal K}^{\times } := {\mathcal K} \setminus \{0 \}$.
Then ${\mathcal K}^{\times } \leq \GL_n(k)$ is isomorphic to the
multiplicative group $K^{\times }$ of $K$ and hence cyclic of order $q^n-1$.
Let $S$ be any generator of ${\mathcal K}^{\times } = \langle S \rangle $ as a group. In group theory $S$ is often called a Singer cycle.
Clearly ${\mathcal K}^{\times }$  contains the subgroup of nonzero scalar matrices
$$C_{q-1} \cong k^{\times } \cong k^{\times } I_n = 
\langle S^{(q^n-1)/(q-1)} \rangle = \langle SS^qS^{(q^2)} \cdots S^{(q^{n-1})} \rangle \leq {\mathcal K} ^{\times } .$$
Furthermore if $X\in  {\mathcal K}^{\times } 
$, then
$X {\mathcal G}_{\ell ,\Gamma }  = {\mathcal G}_{\ell ,\Gamma } $
and
${\mathcal G}_{\ell ,\Gamma }  X = {\mathcal G}_{\ell ,\Gamma } $.
By Lemma \ref{notation} (ii) conjugation by $A$ preserves the set ${\mathcal K}$,
so
$A^j {\mathcal G}_{\ell ,\Gamma } A^{-j} = {\mathcal G}_{\ell ,\Gamma } $ for $j=0,\ldots , n-1$.

The next theorem shows that these obvious automorphisms 
already generate the full automorphism group of the Gabidulin codes.

\begin{theorem}\label{properaut}
For $0< \ell < n$ the group of proper automorphisms of ${\mathcal G}_{\ell ,\Gamma }$ is 
$$\Aut ^{(p)} ({\mathcal G}_{\ell ,\Gamma } ) = 
\{ \kappa _{X,Y}  \mid (X,Y) \in (A^j {\mathcal K}^{\times } \times A^{-j} {\mathcal K}^{\times } ) , 0\leq j \leq n-1 \}  $$
which is isomorphic to the semidirect product of 
$C_n\cong \Gal (K/k) $ with the normal subgroup 
${\mathcal K}^{\times } {\sf Y} {\mathcal K}^{\times } $ 
 the central product of $K^{\times }$ with itself amalgamated 
over $k^{\times }$.
\end{theorem}

\begin{proof} The inclusion $\supseteq$ is clear. To see the converse we suppose that $X {\mathcal G}_{\ell ,\Gamma } Y = {\mathcal G}_{\ell ,\Gamma } $ for
$X,Y \in \GL_n(k)$.
\\[1ex]
Claim 1: 
If ${\mathcal Z}:= \{ Z\in \GL_n(k) \mid Z{\mathcal G}_{\ell ,\Gamma } = {\mathcal G}_{\ell ,\Gamma } \}$ then ${\mathcal Z}  = {\mathcal K}^{\times } $: 
\\
According to Lemma \ref{notation} (v) we may write $Z:= \sum _{i=0}^{n-1} z_i A^i \in {\mathcal Z} $ where
$z_i \in {\mathcal K}$ for all $i$.
As $I_n  \in {\mathcal K}^{\times } \subseteq {\mathcal G}_{\ell ,\Gamma } $ also 
$Z = ZI_n\in {\mathcal G}_{\ell ,\Gamma } $, so  $z_i=0$ for $i=\ell,\ldots, n-1$. 
If $\ell \geq 1$, then  also $A\in {\mathcal G}_{\ell ,\Gamma } $. Thus 
$ZA = \sum _{i=0}^{\ell -1} z_i A^{i+1} \in {\mathcal G}_{\ell ,\Gamma } $,
which implies that $z_{\ell -1} = 0 $.
 Repeating this argument several times  we obtain
$z_1 =\ldots = z_{n-1} = 0$ and $Z=z_0\in {\mathcal K} $.
\\[1ex]
Claim 2:  $X {\mathcal Z} X^{-1} = {\mathcal Z} (={\mathcal K}^{\times })$: \\
$X {\mathcal G}_{\ell ,\Gamma } Y = {\mathcal G}_{\ell ,\Gamma }$
is obviously invariant under left multiplication with $X {\mathcal Z} X^{-1}$. Thus  
 Claim 1 implies $X {\mathcal Z} X^{-1} = {\mathcal Z} $.
\\[1ex]
Final step: 
By Claim 2 we know that
 $X \in \GL_n(k)$ lies in the normalizer of  ${\mathcal K}^{\times }$. 
Note that $A$ induces by conjugation on ${\mathcal K}^{\times } $ the Galois automorphism $x \mapsto x^q$ (cf. Lemma \ref{notation} (ii)).
By Lemma \ref{notation} (iii) the normalizer of ${\mathcal K}^{\times }$ is 
$ N_{\GL_n(k)}({\mathcal K}^{\times }) =  \langle A \rangle {\mathcal K}^{\times }$.
Therefore there is some $0\leq j \leq n-1$ such that 
$X\in A^j {\mathcal K}^{\times }$. 
In particular $X {\mathcal G}_{\ell ,\Gamma } X^{-1} = {\mathcal G}_{\ell ,\Gamma } $ and hence 
${\mathcal G}_{\ell ,\Gamma } XY = {\mathcal G}_{\ell ,\Gamma } $. Similar to the proof of Claim 1 we conclude that
 $XY \in {\mathcal K}^{\times } $, hence $Y \in {\mathcal K}^{\times } A^{-j }{\mathcal K}^{\times } = A^{-j} {\mathcal K}^{\times }$.
\end{proof}

\begin{cor}\label{fullaut}
For $0< \ell < n$ the full automorphism group of ${\mathcal G}_{\ell ,\Gamma }$ is 
$$\Aut ({\mathcal G}_{\ell , \Gamma } ) = 
\langle \Aut ^{(p)} ({\mathcal G}_{\ell ,\Gamma } ) , \tau _{\Tau ^{-1},\Tau A^{\ell-1}} \rangle $$
and contains the group of proper automorphisms from Theorem \ref{properaut}
of index 2.
In particular 
$$|\Aut ({\mathcal G}_{\ell , \Gamma } ) | = 2n(q^n-1) \frac{q^n-1}{q-1} .$$
\end{cor}

\begin{proof}
For any subgroup $U\leq G$ of some finite group $G$ and a normal subgroup $N\unlhd G$,
we have $|U/(N\cap U) | \leq |G/N |$. 
So in particular the index of $\Aut ^{(p)}  ({\mathcal G}_{\ell ,\Gamma } ) $ in 
the full automorphism group is either $1$ or $2$ and it suffices to 
show that $\tau _{\Tau ^{-1},\Tau A^{\ell -1 }}  ({\mathcal G}_{\ell ,\Gamma } ) = 
{\mathcal G}_{\ell ,\Gamma } $.
To this aim let $C\in {\mathcal K}$ and $0\leq j\leq \ell -1 $.
Then 
$$ \begin{array}{rcl} \tau _{\Tau ^{-1},\Tau A^{\ell -1 }} (CA^j) & = &
\Tau ^{-1} (CA^j)^{\tr} \Tau  A^{\ell -1}  = \Tau ^{-1} (C')^{\tr} A^{-j} \Tau  A^{\ell -1} \\[1ex] & = &
\Tau ^{-1} (C')^{\tr} \Tau  A^{\ell -1 -j } \end{array}$$
for some $C'\in {\mathcal K}$, because conjugation by $A$ preserves ${\mathcal K}$ as 
a set. The last equality follows from Remark \ref{vert} (i). 
By Lemma \ref{basechange}, we have $\Tau ^{-1}  (C')^{\tr} \Tau = C' \in {\mathcal K}$. 
So $\tau _{\Tau ^{-1},\Tau A^{\ell -1} } $ maps ${\mathcal K} A^j $
onto ${\mathcal K} A^{\ell -1 -j }$ and hence preserves the code ${\mathcal G}_{\ell ,\Gamma } $.
\end{proof}

\subsection{Self-dual Gabidulin codes.} 

According to Theorem \ref{char2} and the fact that Gabidulin codes are $\MRD$ codes,
there are no self-dual Gabidulin codes in even characteristic.
So in this section we assume that $k=\F_{q}$, $K := \F_{q^n}$ and $q$ is odd. 
We keep the notation from above. In 
particular $\Gamma = (\gamma,\gamma^{[1]},\ldots , \gamma ^{[n-1]}$ is 
a  normal basis of $K/k$, 
$\Tau :=\Tau _{\Gamma }$, $\langle S \rangle = {\mathcal K}^{\times }$,
and $A:=\epsilon_{\Gamma }(\Gamma ^{[1]}) $.
%
If the Gabidulin code ${\mathcal G}_{\ell ,\Gamma } $ 
is equivalent to a self-dual code then  $\ell =n/2$ and
$n$ needs to be even. 
The following facts are elementary but crucial for the proofs of Proposition
\ref{dualmat} and Theorem \ref{gabisd} below.

\begin{lemma}\label{basic}
Assume that $n$ is even. Then
\begin{itemize}
\item[\rm (i)] $\det(A) = -1 $ and $A^{\tr} = A^{-1}$.
\item[\rm (ii)] $A S A^{-1} = S^q$.
\item[\rm  (iii)] $A \Tau  = \Tau A $.
In particular $(\Tau A^{\ell})^{\tr} = A^{-\ell}\Tau  $.
\item[\rm  (iv)] $\det(S) $ is a primitive element of $\F_q$.
\item[\rm  (v)] $\det(\Tau ) \not\in (\F_q^{\times })^2 $.
\item[\rm  (vi)]
$\Tau S^j$  is symmetric for all $j=0,\ldots , q^n-1$.
\item[\rm  (vii)] $S^j\Tau ^{-1}$ is symmetric for all $j=0,\ldots , q^n-1$.
\item[\rm  (viii)]
$(\Tau A^jS^i)$ is symmetric if and only if
$\left\{ \begin{array}{l} j=n/2 \mbox{ and }  (q^{n/2} +1) \mid i \\ \mbox{ or } \\
j=0 \mbox{ and } i \in \{ 1,\ldots , q^n-1 \} , \end{array} \right. $ \\
if and only if $S^i A^j \Tau ^{-1} $ is symmetric.
\end{itemize}
\end{lemma}

\begin{proof}
(i) This is clear as $A$ is a permutation matrix of a cycle of full length $n$ and
$n$ is even.
\\[1ex]
(ii)  This follows from Lemma \ref{notation} (ii).
\\[1ex]
(iii) The first statement is Remark \ref{vert} (i). To see the second
note that $A$ is a permutation matrix, so $A^{\tr} = A^{-1}$  and
$(\Tau A^{\ell})^{\tr} = (A^{\ell})^{\tr} \Tau   = A^{-\ell } \Tau  $.
\\[1ex]
(iv) Because $S$ generates ${\mathcal K}$ as a $k$-algebra
the minimal polynomial of $S$
is equal to its characteristic polynomial. Moreover it also coincides with the minimal polynomial
of a primitive element $\sigma \in \F_{q^n}$ over $\F_q$ since $S$ is a Singer cycle.
 Thus the determinant of $S$ is
the product of all Galois conjugates of $\sigma $, i.e. the norm of $\sigma $,
$$\det(S) = \sigma ^{(1+q+\ldots + q^{n-1})} = \sigma ^{(q^n-1)/(q-1)} .$$
As $\langle \sigma \rangle = \F_{q^n}^{\times }$ the order of $\sigma $ is $q^n-1$, so
the order of $\det (S)$ is $q-1$ which proves that $\det(S) $ is a primitive
element  in $\F_q$. In particular $\det(S) \in \F_q^{\times }\setminus (\F_q^{\times })^2$.
\\[1ex]
(v)
By Lemma \ref{qf}, there is a self-dual basis for $K/k$ if and only
if the determinant of the trace bilinear form is a square.
According to Lempel and Seroussi  \cite{LS} $K/k$ has
a self-dual basis if and only if $n$ is odd (since $q$ is odd).
As $n$ is assumed to be even,
 the determinant of $\Tau $ is a non-square.
\\[1ex]
(vi)  Lemma \ref{basechange} with $\epsilon _{\Beta }(\alpha \Beta )=S$ implies that $S^{\tr} = \Tau  S \Tau ^{-1}$, hence
$$(\Tau S)^{\tr}  = S^{\tr} \Tau ^{\tr} =  (\Tau  S \Tau ^{-1}) \Tau  = \Tau S .$$
(vii) This follows from (vi) because the inverse of a symmetric matrix is again symmetric.
\\[1ex]
(viii) Using the previous results we compute
$$(\Tau A^jS^i) ^{\tr} \stackrel{(ii)}{=} (\Tau S^{q^ji} A^j)^{\tr} \stackrel{(vi)}{=} A^{-j} \Tau  S^{q^ji}  
\stackrel{(iii)}{=} \Tau  A^{-j} S^{q^ji} \stackrel{(ii)}{=} \Tau  S^i A^{-j} $$
for all $0\leq j \leq n-1 $ and $1\leq i \leq q^n-1 $.
In particular $\Tau A^jS^i$ is symmetric if and only if
$\Tau A^jS^i = \Tau  S^i A^{-j} $. Dividing by $\Tau $ and using 2) we obtain the 
equivalent condition
$ S^{q^ji} A^j = S^i A^{-j} \in {\mathcal K }A^j \cap {\mathcal K}A^{-j} $.
Now ${\mathcal K} A^r \cap {\mathcal K} A^s \neq \{ 0 \} $ if and only if $r \equiv s \bmod n$.
So we obtain that $j \equiv -j \bmod n$, i.e. either $j=0$ and then $i$ is arbitrary,
or $j=n/2$ and $(S^i)^{q^{n/2} } = (S^i)$ (i.e. $ (q^{n/2} +1) \mid i  $).
The last statement follows by inverting the matrix.
\end{proof}

The next proposition follows by interpreting 
\cite[Lemma 1]{Berger} and \cite[Theorem 18]{Ravagnani} in our language. 
For convenience of the reader we give a direct elementary proof.

\begin{proposition}\label{dualmat}
${\mathcal G}_{n/2 ,\Gamma } ^{\perp} =
\Tau  A^{n/2} {\mathcal G}_{n/2 ,\Gamma } \Tau  ^{-1} $.
\end{proposition}

\begin{proof}
We put ${\mathcal C}:=  \Tau  A^{n/2} {\mathcal G}_{n/2 ,\Gamma } \Tau  ^{-1} $.
As $$\dim({\mathcal C}) + \dim ({\mathcal G}_{n/2 ,\Gamma }) 
= 2 \frac{n}{2} n = n^2 = \dim (k^{n\times n})$$
it suffices to show that ${\mathcal C} \subseteq {\mathcal G}_{n/2 ,\Gamma } ^{\perp} $.
To see this recall that
${\mathcal G}_{n/2,\Gamma } = \bigoplus _{i=0}^{n/2-1} {\mathcal K} A^i $
where ${\mathcal K} = \{ 0 \} \cup \{ S^{\ell } \mid 0\leq \ell \leq q^n-1 \} $
and ${\mathcal K} A = A {\mathcal K}$.
So it is enough to show that for $i\neq j$ with $i,j \in \{ 0,\ldots , n-1 \}$
and all $m,\ell \in \{ 0,\ldots , q^n-1  \} $
$$\trace (S^mA^i (\Tau S^{\ell} A^j \Tau ^{-1} )^{\tr } ) = 0 .$$
Applying Lemma \ref{basic} (where the relevant parts are
indicated above the equalities) we compute
$$S^mA^i (\Tau S^{\ell} A^j \Tau ^{-1} )^{\tr } \stackrel{(i),(vi)}{=} 
S^mA^i \Tau ^{-1} A^{-j} \Tau S^{\ell } \stackrel{(iii)}{=} S^m A^{i-j} S^{\ell} \, {\in} \, {\mathcal K} A^{i-j} $$
where the last inclusion follows from Lemma \ref{basic} (ii).
If $i-j$ is not divisible by $n$, then Lemma \ref{notation} (iii) tells us that
all matrices in  $ {\mathcal K} A^{i-j}  $ have trace $0$.
\end{proof}

In particular ${\mathcal G}_{n/2 ,\Gamma } $ is always equivalent
to its dual code. 
We now apply Theorem \ref{eqsd} to obtain a criterion,
when ${\mathcal G}_{n/2 ,\Gamma }$ is equivalent to a self-dual code. 

\begin{theorem}\label{gabisd}
${\mathcal G}_{n/2 ,\Gamma }$ is equivalent to a self-dual $\MRD$ code 
if and only if  $n\equiv 2 \pmod{4}$ and $q\equiv 3 \pmod{4}$. 
\end{theorem}

\begin{proof}
Let $\delta := \det(S)$. Then, by Lemma \ref{basic} (iv),
$\delta \not\in (k^{\times })^2$.  \\[1ex]
By Proposition \ref{dualmat}, we have
$$
{\mathcal G}_{n/2 ,\Gamma } ^{\perp} =
\Tau  A^{n/2} {\mathcal G}_{n/2 ,\Gamma }  \Tau^{-1} .$$
From Theorem \ref{properaut} we hence obtain the set of proper equivalences
between ${\mathcal G}_{n/2,\Gamma }$ and ${\mathcal G}_{n/2 ,\Gamma } ^{\perp} $ as
$$\{ \kappa _{\Tau  A^{n/2} A^j S^i , S^h A^{-j}  \Tau^{-1} } \mid i,h \in \{ 0,\ldots , q^n-1 \} , j\in \{ 0,\ldots , n-1 \} \} .$$
According to Corollary \ref{fullaut}  all Gabidulin codes have
improper automorphisms. So if ${\mathcal G}_{n/2 ,\Gamma }$ is equivalent to a self-dual $\MRD $ code, then it is properly equivalent to
a self-dual $\MRD $ code. \\[1ex]
To use  Theorem \ref{eqsd} we hence need to decide for which triples $(i,h,j)$ both matrices
$$X_{i,j} := \Tau  A^{n/2} A^j S^i \mbox{ and } Y_{h,j} := S^h A^{-j}  \Tau^{-1} $$
are symmetric and of square determinant.
\\[1ex]
 By Lemma \ref{basic} (v) (note that we assume that $n$ is even),
$\det (X_{i,j}) \in (k^{\times })^2$ if and only if $(-1) ^{\frac{n}{2} + j} \delta^i \not\in (k^{\times })^2 $ and
$\det (Y_{h,j}) \in (k^{\times })^2$ if and only if $(-1)^{j} \delta ^h \not\in (k^{\times })^2$.
\\[1ex]
 By Lemma \ref{basic} (viii), the matrix $X_{i,j}$ is symmetric if and only if either $j=0$ and $(q^{n/2} +1) \mid i$
or $j=n/2$ and $i$ is arbitrary.
The matrix $Y_{h,j}$ is symmetric if and only if either $j=0$ and $h$ is arbitrary or $j=n/2$ and $h$ is a multiple of $(q^{n/2} +1) $.
\\[1ex]
 So  in particular $X_{i,j}$ and $Y_{h,j}$ are symmetric of square determinant  if and only if
either \begin{itemize}
\item[(a)] $(-1)^{n/2} \not\in (k^{\times })^2$ and $j=0, (q^{n/2}+1) \mid i,$ and $h$ is odd
\item[or]
\item[(b)]  $(-1)^{n/2} \not\in (k^{\times })^2$ and $j=n/2$, $(q^{n/2}+1) \mid  h$, and $i$ is odd.
\end{itemize}
These conditions can be satisfied
if and only if  $n\equiv 2 \pmod{4}$ and $q\equiv 3 \pmod{4}$.
\end{proof}

\noindent
{\bf Acknowledgement.} 
The ideas for this paper initiated during the ALCOMA15 conference.
The authors thank the organisers for their kind invitation. 
Part of the work was done during two visits of the second author to
the RWTH Aachen University in spring 2015 
financed by the RTG 1632 of the DFG.

\end{document}